%% file: main.tex
\documentclass{article}
\usepackage{graphicx} 

\usepackage{amsmath,amsthm,amssymb,amsfonts}
\usepackage{caption}
\usepackage{subcaption}
\usepackage{bbm}
\usepackage{graphicx}
\usepackage{multirow}
\usepackage{xspace}
\usepackage[colorlinks=true,citecolor=blue,linkcolor=blue]{hyperref}
\usepackage{cleveref}
\usepackage{siunitx,booktabs}
\usepackage{algorithmic}

\usepackage[letterpaper,top=2cm,bottom=2cm,left=3cm,right=3cm,marginparwidth=1.75cm]{geometry}

\usepackage{fancyhdr}
\usepackage{enumerate}
\usepackage{tikz}
\usetikzlibrary{arrows,shapes,automata, calc, chains, matrix, trees, positioning, scopes}
\usetikzlibrary{decorations.pathmorphing}
\tikzset{snake it/.style={decorate, decoration=snake}}

\newtheorem{theorem}{Theorem}
\newtheorem{proposition}[theorem]{Proposition}

\newtheorem{lemma}[theorem]{Lemma}
\newtheorem{claim}[theorem]{Claim}
\newtheorem{remark}[theorem]{Remark}

\newtheorem*{problem*}{Problem}

\numberwithin{theorem}{section}

\theoremstyle{definition}

\newcommand{\comment}[1]{}

\newcommand{\OO}{\mathcal{O}}
\newcommand{\Aa}{\mathcal{A}}

\newcommand{\MM}{\mathcal{M}}

\newcommand{\ACZ}{\textnormal{\textsf{AC$^0$}}\xspace}

\newcommand{\DET}{\textnormal{\textsf{DET}}\xspace}

\newcommand{\PSPACE}{\textnormal{\textsf{PSPACE}}\xspace}

\newcommand{\NL}
{\textnormal{\textsf{NL}}\xspace}
\newcommand{\coNL}
{\textnormal{\textsf{coNL}}\xspace}
\newcommand{\NCT}
{\textnormal{\textsf{NC}$^2$}\xspace}

\newcommand{\LL}{\textnormal{\textsf{L}}\xspace}
\newcommand{\PTIME}{\textnormal{\textsf{P}\xspace}}

\title{The complexity of reachability problems \\ in strongly connected finite automata}
\author{Stefan Kiefer$^1$, Andrew Ryzhikov$^2$}
\date{$^1$ Department of Computer Science, University of Oxford, UK \\ $^2$ University of Warsaw, Poland}

\begin{document}

\maketitle

\begin{abstract}
Several reachability problems in finite automata, such as completeness of NFAs and synchronisation of total DFAs, correspond to fundamental properties of sets of nonnegative matrices. In particular, the two mentioned properties correspond to matrix mortality and ergodicity, which ask whether there exists a product of the input matrices that is equal to, respectively, the zero matrix and a matrix with a column of strictly positive entries only. The case where the input automaton is strongly connected (that is, the corresponding set of nonnegative matrices is irreducible) frequently appears in applications and often admits better properties than the general case. In this paper, we address the existence of such properties from the computational complexity point of view, and develop a versatile technique to show that several \NL-complete problems remain \NL-complete in the strongly connected case. In particular, we show that deciding if a binary total DFA is synchronising is~\NL-complete even if it is promised to be strongly connected, and that deciding completeness of a binary unambiguous NFA with very limited nondeterminism is~\NL-complete under the same promise.

\vspace{\baselineskip}
    
    \end{abstract}

\section{Introduction}\label{sec:intro-new}
\input{sec-intro-new}

\section{Motivation and initial examples}\label{sec:motivation}
\input{sec-motivation}

\section{Existing results and our contributions}\label{sec:existing}

\input{sec-existing}

\section{Constrained \texorpdfstring{$(s, t)$}{(s, t)}-reachability in digraphs}\label{sec:constr-reach}
\input{sec-constr-reach}

\section{Synchronisation of total DFAs}\label{sec:nl-sync}
\input{sec-nl-sync}

\section{Completeness and unambiguity of NFAs}\label{sec:comp-and-unamb}

In this section, we prove two results that, in a way, complement each other. Given a binary strongly connected NFA, in~\Cref{sec:nl-mortality} we show that it is \NL-hard to decide if it is complete, even if it is additionally promised to be unambiguous. Recall that for unambiguous NFAs this problem is known to be in \DET, a class slightly above \NL. Similarly, in~\Cref{sec:nl-unamb} we show that given a binary strongly connected NFA, it is \NL-hard (and thus \NL-complete) to decide if it is unambiguous, even if it is additionally promised to be complete. In the next section we will see that both results hold true even for a very restricted class of NFAs, for which both problems are in fact~\NL-complete.

\subsection{Completeness of unambiguous NFAs}\label{sec:nl-mortality}
\input{sec-nl-mortality}

\subsection{Unambiguity of complete NFAs}\label{sec:nl-unamb}
\input{sec-nl-unamb}

\section{A restricted class of NFAs}\label{sec:2-im-b}
\input{sec-2-im-b}

\section{Applications to codes}\label{sec:codes}

\input{sec-codes}

\section{Conclusions and open problems}\label{sec:conclusions}
\input{sec-conclusions}

\subsection*{Acknowledgements}

Andrew Ryzhikov is supported by Polish National Science Centre SONATA BIS-12 grant
number 2022/46/E/ST6/00230.

\bibliographystyle{alpha}
\bibliography{sync}

\end{document}

%% file: sec-intro-new.tex
A square nonnegative matrix is called irreducible if permuting its rows and columns cannot result in a matrix of the shape $\begin{pmatrix}A & B \\ 0 & C\end{pmatrix}$, where $A$ and $C$ are square matrices.
In the theory of nonnegative matrices, irreducibility plays a special role. 
For example, for an irreducible nonnegative matrix, the Perron-Frobenius theory guarantees the existence of a strictly positive eigenvector, while for a general nonnegative matrix it only guarantees the existence of a nonnegative one~\cite{Minc1988}, which has important applications in the algebraic theory of digraphs~\cite{Brualdi2010}. As another example, for irreducible nonnegative matrices there exists an easy characterisation of primitive matrices (matrices having a power with only strictly positive entries) in terms of a partition of the set of rows~\cite[Theorem 10.5.1]{BangJensen2008}. As a consequence, primitivity of irreducible nonnegative matrices is~\LL-complete, whereas primitivity of general nonnegative matrices is \NL-complete~\cite{Kiefer2026}.

A finite set of nonnegative matrices is called irreducible if the sum of all its matrices is irreducible. Similarly to the case of one matrix, irreducible sets are much more well-behaved, both in terms of their algorithmic properties and structural characterisations. For example, given a set of nonnegative matrices such that each row contains a single entry equal to~1 and all other entries are zero, deciding if there exists a product of these matrices having (linear) rank one is \PSPACE-complete in general~\cite{Berlinkov2014}, but in the irreducible case it admits an easy characterisation in terms of pairs of rows, and thus a polynomial-time algorithm~\cite{Berlinkov2021}. Further examples of characterisations that are specific to the irreducible case can be found, for example, in~\cite{Protasov2012,Protasov2013,Wu2023}.
Moreover, in several applications, for example in the case of matrix semigroups associated to regular codes~\cite{Berstel2010} and in testing certain classes of reactive systems~\cite{Schramme2016}, sets of nonnegative matrices are guaranteed to be irreducible, hence stronger results immediately apply. 

In this paper we study the consequences of irreducibility on some fundamental properties of sets of nonnegative matrices and finite automata. More concretely, we take a computational complexity view in exploring to what extent irreducibility leads to more efficient characterisations of mortality, ergodicity and finiteness. For finite automata, these three properties correspond to being incomplete, synchronising and unambiguous respectively.
In the next two sections, we formally introduce the properties studied in this paper and provide very simple examples for when a drop in complexity does and does not happen. We then discuss our main contributions and the outline of the paper.

%% file: sec-motivation.tex
\subsection{Properties of finite automata and matrix sets} 

\subparagraph*{From sets of nonnegative matrices to (semi-)automata.}
Let $\MM = \{A_1, \ldots, A_m\}$ be a set of $n \times n$ matrices with nonnegative entries (called \emph{nonnegative matrices}). \emph{The monoid generated by~$\MM$} is the set of all products of matrices from $\MM$. Many combinatorial properties of this monoid (for example, completeness and ergodicity defined below) can be conveniently formulated in terms of reachability properties of the associated nondeterministic finite \mbox{(semi-)automaton} (NFA) $\Aa$. Given $\MM$, the NFA $\Aa = (Q, \Sigma, \Delta)$ is constructed as follows. The set $Q$ of states is $\{1, \ldots, n\}$. For each matrix~$A_i \in \MM$, $1 \le i \le m$, the alphabet~$\Sigma$ contains a corresponding letter $a_i$. The transition relation $\Delta \subseteq Q \times \Sigma \times Q$ is defined as follows: $(i, a_k, j) \in \Delta$ if and only if the entry $(i, j)$ in $A_k$ is strictly positive. An example of this construction is illustrated in~\Cref{fig:diamond} (left and centre). We emphasise that our definition of an NFA does not include any initial or final states, and thus simply represents an edge-labelled directed graph, similarly to those used in symbolic dynamics.

\begin{figure}[ht]\centering
\begin{subfigure}[c]{0.30\textwidth} \[\Biggl\{\begin{pmatrix}
    2 & 1 & 0 \\
    0 & 0 & 0 \\
    4 & 0 & 0
\end{pmatrix},
\begin{pmatrix}
    0 & 0 & 1 \\
    0 & 0 & 7 \\
    0 & 0 & 0
\end{pmatrix}\Biggl\}\]
\end{subfigure}
\quad
\begin{subfigure}[c]{0.59\textwidth}\centering
\begin{tikzpicture} [node distance = 2cm]
\tikzset{every state/.style={inner sep=1pt,minimum size=1.5em}}

\node [state] at (0, 0) (1) {1};
\node [state] at (1.5, 1.5) (2) {2};
\node [state] at (3, 0) (3) {3};

\path [-stealth, thick]
(1) edge [] node[above] {$a$} (2)
(2) edge [] node[above] {$b$} (3)
(1) edge [loop below] node[below] {$a$} (1)
(1) edge [] node[above] {$b$} (3)

;

\path [-stealth, thick]
(3) edge [bend left=30] node[below] {$a$} (1)
;
\end{tikzpicture}
\quad
\begin{tikzpicture} [node distance = 2cm]
\tikzset{every state/.style={inner sep=1pt,minimum size=1.5em}}

\node [state] at (-1.75, 0) (p) {$p$};
\node [state] at (0, 0.75) (t1) {$t_1$};
\node [state] at (1.75, 0) (q) {$q$};
\node [state] at (0, -0.75) (t2) {$t_2$};
\path [-stealth, thick]

(p) edge [decorate, decoration={snake, segment length=3mm, amplitude=0.5mm}] node[above] {$w_1$} (t1)
(t1) edge [decorate, decoration={snake, segment length=3mm, amplitude=0.5mm}] node[above] {$w_2$} (q)

(p) edge [decorate, decoration={snake, segment length=3mm, amplitude=0.5mm}] node[below] {$w_1$} (t2)
(t2) edge [decorate, decoration={snake, segment length=3mm, amplitude=0.5mm}] node[below] {$w_2$} (q)
;
\end{tikzpicture}
\end{subfigure}
\caption{A set of two nonnegative matrices (left), the corresponding NFA (centre) and an illustration of a diamond (right). By taking $p = 1$, $w_1 = a$, $w_2 = b$ and $q = 3$, one can see a diamond in the NFA.}\label{fig:diamond}
\end{figure}
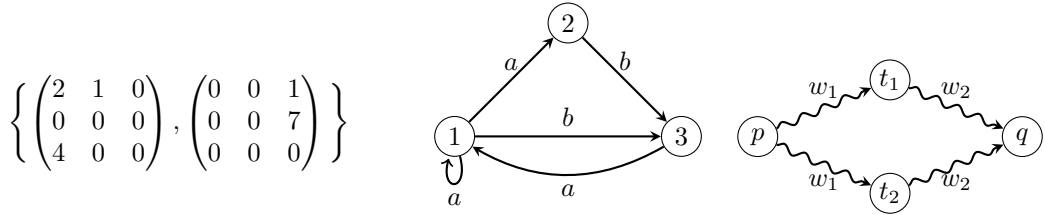

An NFA $\Aa = (Q, \Sigma, \Delta)$ is a \emph{deterministic finite (semi-)automaton (DFA)} if for every state $p \in Q$ and every letter $a \in \Sigma$ there is at most one outgoing transition, that is, there is at most one state $q$ with $(p, a, q) \in \Delta$.  If there is exactly one such state $p$ for every $q$ and $a$, a DFA is called \emph{total}.

There are straightforward connections between properties of the matrix set $\MM$ and the corresponding NFA~$\Aa$:
\begin{itemize}
    \item An NFA is called \emph{complete} if every word in $\Sigma^*$ labels a path in $\Aa$. $\Aa$ is complete if and only if the monoid generated by $\MM$ does not contain the zero matrix. Deciding the latter property is known as the the matrix mortality problem~\cite{Paterson1970}, and is a subject of active research, see e.g.~\cite{Cassaigne2014,Bell2021,Kiefer2021,Ryzhikov2024RP}. If one considers a linear dynamical system where a matrix from $\MM$ is applied at every step, mortality simply means that there is a scenario where the system is mapped to the zero configuration regardless of where it started from.

    \item An NFA is called \emph{strongly connected} if its underlying digraph (where we forget the labels of the transitions) is strongly connected. NFA $\Aa$ is strongly connected if and only if $\MM$ is irreducible, that is, if for every pair $(i, j)$, $1 \le i, j \le n$, the monoid generated by $\MM$ contains a matrix whose entry $(i, j)$ is strictly positive. 

    \item $\Aa$ is a DFA if and only if every matrix in the set $\MM$ contains at most one strictly positive entry in each row.

    \item A total DFA is called \emph{synchronising} if there exists a word mapping every state to the same state. Total DFA $\Aa$ is synchronising if and only if the monoid generated by $\MM$ contains a matrix with a column of strictly positive entries only ($\MM$ is sometimes called ergodic in this case~\cite{Wu2023}). Intuitively, synchronisation means that a total DFA admits a word that ``resets'' it to a particular state regardless of its current state.

\end{itemize}

\subparagraph*{Characterisations of properties and relations with~\NL-hardness.} 
As mentioned in the introduction, irreducibility is particularly important. 
In some applications (such as deciding if a regular code is complete or synchronising~\cite{Berstel2010,Berlinkov2021}, see \Cref{sec:codes} for further discussion), the constructed NFA is always strongly connected. 
Besides that, deciding many other reachability properties (such as completeness or unambiguity of NFAs) often reduces to the analysis of strongly connected components of an NFA. Hence, to understand such properties better, it is useful to consider them specifically for strongly connected NFAs. Since deciding strong connectivity is \NL-complete~\cite[Problem~8.27]{Sipser2013}, this is especially relevant for properties that can be decided in~\NL. In some cases, restricting to strongly connected NFAs results in a drop of complexity, indicating that there might be a better characterisation of the property for strongly connected NFAs in comparison to general NFAs. In other cases, the complexity remains the same, showing that the property in question already expresses strong connectivity in an implicit way and every characterisation of it must also be \NL-hard to decide. 

\subsection{Two simple examples} \label{subs:two-examples}

Let us provide a very simple illustration of both scenarios: where the complexity drops in the strongly connected case, and where it remains the same but requires a more involved proof. Both examples involve well-known problems for DFAs, however the simple proofs that we present in this section seem to be new.

One of the most fundamental \NL-complete problems is the $(s, t)$-reachability problem in digraphs \cite[Theorem~8.25]{Sipser2013}. Given a directed graph $G = (V, E)$ and two vertices $s, t \in V$, this problem asks if there exists an $(s, t)$-path in $G$ (that is, a path starting in $s$ and ending in $t$). It is well-known that this problem is \NL-complete for acyclic digraphs where all vertices have outdegree at most two (we recall the techniques for proving that in the proof of~\Cref{lem:st-reach-ass}), so we assume that $G$ is such a digraph. We can also assume that vertices of indegree zero have outdegree one, and that $t$ has outdegree zero.

\subparagraph*{Completeness of DFAs.} 
Consider first the problem of deciding if a given DFA is complete. It is \NL-hard by the following simple reduction from $(s, t)$-reachability in digraphs. Let $G = (V, E)$ and $s, t \in V$ be its input.
The reduction works as follows. Merge all vertices of outdegree zero except for~$t$, and denote the resulting vertex~$t'$ (since $G$ is acyclic, if there are no such vertices, then an $(s, t)$-path always exists). Add an edge $(t', s)$. For every vertex $s'$ of indegree zero except for~$s$, add an edge $(s', t)$. Label all edges of the resulting digraph with letters $a$ and $b$ in such a way that every vertex different to~$t$ has exactly one outgoing edge labelled by each of $a$ and $b$ (duplicate the edges where necessary). This construction is illustrated in \Cref{fig:intro-nl-constr}. 

\begin{figure}[ht]\centering
\begin{tikzpicture} [node distance = 2cm,baseline=(s)]
\tikzset{every state/.style={inner sep=1pt,minimum size=1.5em}}

\node [state] at (-1.5, 1) (s) {$s$};
\node [state] at (-1.5, -0.5) (2) {};

\node [state] at (0, 1.5) (1) {};

\node [state] at (1.5, 1.5) (3) {};

\node [state] at (3, 1) (t) {$t$};
\node [state] at (3, 0) (4) {};

\path [-stealth, thick]
(s) edge [] node[above] {} (1)
(1) edge [] node[above] {} (3)
(2) edge [] node[above] {} (4)
(1) edge [] node[above] {} (4)

;

\end{tikzpicture}
\qquad
\begin{tikzpicture} [node distance = 2cm,baseline=(s)]
\tikzset{every state/.style={inner sep=1pt,minimum size=1.5em}}

\node [state] at (-1.5, 1) (s) {$s$};
\node [state] at (-1.5, -0.5) (2) {};

\node [state] at (0, 1.5) (1) {};

\node [state] at (3, 1) (t) {$t$};
\node [state] at (3, 0) (4) {$t'$};

\path [-stealth, thick]
(s) edge [] node[above] {} (1)
(s) edge [dashed] node[above] {} (4)
(2) edge [] node[above] {} (t)
(2) edge [dashed] node[above] {} (4)

(1) edge [bend left=20] node[above] {} (4)
(1) edge [bend left=30, dashed] node[above] {} (4)

(4) edge [bend left=40] node[above] {} (s)
(4) edge [bend left=50, dashed] node[above] {} (s)

;

\end{tikzpicture}
\caption{The input digraph (left) and the resulting DFA in the reduction (right). Transitions by $a$ are represented by solid arrows, and transitions by $b$ by dashed lines.}\label{fig:intro-nl-constr}
\end{figure}
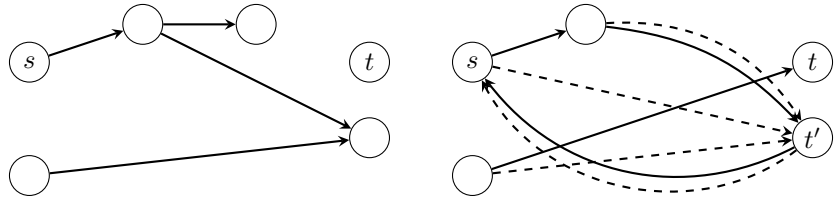

Put a token on every state of the resulting DFA and move them according to a word $w$ that is read letter by letter. If, upon reading a letter~$x \in \{a, b\}$, there is a token on a state with no outgoing transition from this state labelled by $x$, remove the token. By definition, a DFA is not complete if and only if there is a word removing all the tokens.
Observe that $t$~is the only state where tokens can be removed. If every state can be mapped to~$t$, then we can subsequently move the tokens to~$t$ one by one and remove them. Conversely, if there is a state that cannot be mapped to $t$ by any word, the DFA is complete.
It remains to note that by construction every state can be mapped to either $t$ or $t'$, and $t'$ can be mapped to~$s$. Hence, every state can be mapped to $t$ if and only if there is an $(s, t)$-path in $G$, which concludes the proof.

It is easy to see that deciding if a DFA is incomplete is in~\NL (by checking if every state can be mapped to a state where a token can be removed). Since $\NL=\coNL$~\cite[Section~8.6]{Sipser2013}, deciding if a DFA is complete is \NL-complete.

Intuitively, the reason why the described construction for showing \NL-hardness works is that it guarantees the following two conditions on the reachability relation between states. On the one hand, state $t$ is ``special'' and serves as a trivial state-to-state reachability obstacle: the only reason for the DFA to be complete is that $t$ cannot be reached from a subset of states. On the other hand, the states in this subset can be made pairwise reachable, hence any of them can be taken as the source state. 
This allows us to reduce deciding state-to-state reachability to deciding incompleteness, a reachability property on sets of states.
This is no longer possible if the input DFA is promised to be strongly connected. Indeed, deciding if a strongly connected DFA is complete is in~$\ACZ$, a class significantly below \NL: one simply needs to check if each state has an outgoing transition for every letter. Informally, this is due to the fact that we can no longer express a complex enough pairwise reachability relation.

\subparagraph*{Synchronisation of total DFAs.} 
Consider now the problem of deciding if a total DFA is synchronising. Virtually the same construction as above shows that it is \NL-hard. Indeed, just add self-loops labelled by~$a$ and $b$ to~$t$. By the same argument as above, the resulting total DFA is synchronising if and only if there exists an $(s, t)$-path in $G$. 
Previously, \NL-hardness was already shown in \cite[Lemma 14]{Holzer2018} and \cite[Proposition 2.2]{Volkov2022}, but our construction provides a simpler proof of this fact.

It is well-known that a total DFA is synchronising if and only if for every pair of states there exists a word mapping them to the same state~\cite{Volkov2008}, which also implies that deciding if a total DFA is synchronising is in~\NL. Synchronisability can thus be seen as a reachability condition on pairs of states, and, similarly to completeness, the described construction turns synchronisability into a digraph reachability problem. In particular, the total DFAs constructed for no-instances of $(s, t)$-reachability both in our proof and in the proofs in \cite{Holzer2018,Volkov2022} are not synchronising for trivial reachability reasons: the strongly connected component of~$s$ and the strongly connected component of~$t$ cannot be reached from each other and there is no strongly connected component that they both can reach.

One could expect that, similarly to completeness, deciding whether a total DFA is synchronising becomes easier if it is promised to be strongly connected. Moreover, observe that if a total DFA is synchronising, the period (that is, the greatest common divisor of the lengths of the cycles) of its underlying digraph must be equal to one. In~\cite{Kiefer2026}, it is shown that deciding if the period of a digraph is equal to one is~\NL-complete in general, but becomes~\LL-complete if the digraph is promised to be strongly connected.  
One might hope that synchronisability of total DFAs has a similarly nice and efficient characterisation once strong connectedness is assumed.

However, in~\Cref{sec:nl-sync} we show that deciding synchronisability remains \NL-complete for strongly connected total DFAs. As already mentioned when discussing completeness, the strongly connected case no longer allows for simple reachability obstacles, and makes reachability properties more ``symmetric'' and ``uniform'' between different states and sets of states. Hence, it requires developing a significantly more involved proof technique than the one described above.
It turns out that this proof technique, which is the main technical contribution of the paper, is quite powerful and versatile: it allows us to show that several other fundamental reachability problems are \NL-hard for strongly connected NFAs, even in very restricted settings. We then use these results and their proofs as a guidance for identifying structurally interesting classes of NFAs (and thus of sets of nonnegative matrices).
In the next section we provide the remaining definitions, explain existing results and outline our contributions.

%% file: sec-existing.tex
The remainder of the paper is formulated in terms of NFAs, and we refer to sets of nonnegative matrices only when explaining applications and equivalent formulations for them. For the sake of convenience, this section also repeats the definitions for NFAs that were introduced in the previous section.

\subparagraph*{NFAs and DFAs.}
An NFA $\Aa = (Q, \Sigma, \Delta)$ consists of a finite set $Q$ of states, a finite alphabet~$\Sigma$ and a transition relation $\Delta \subseteq Q \times \Sigma \times Q$. If $|\Sigma| = 2$, $\Aa$ is called \emph{binary}. The relation $\Delta$ is naturally extended to a relation on $Q \times \Sigma^* \times Q$, which we also denote as $\Delta$. For a state $q \in Q$ and a word $w \in \Sigma^*$, we denote by $q \cdot w$ the image of $q$ by $w$, that is, the set $\{p \mid (q, w, p) \in \Delta\}$. We say that $w$ \emph{maps} $q$ to $q \cdot w$. If $q \cdot w = \emptyset$, we say that $w$ \emph{kills} $q$.
An NFA is called \emph{strongly connected} if for all $p,q \in Q$ there is $w \in \Sigma^*$ with $p \cdot w \ni q$. 

An NFA $\Aa = (Q, \Sigma, \Delta)$ is a \emph{deterministic finite (semi-)automaton (DFA)} if for every state $q \in Q$ and every letter $x \in \Sigma$ we have $|q \cdot x| \le 1$. In this case, we use a partial transition function $\delta \colon Q \times \Sigma \rightharpoonup Q$ instead of $\Delta$. It is then also naturally extended to $Q \times \Sigma^* \rightharpoonup Q$, which we still denote as $\delta$. If $\delta$ is a total function, that is, if it is defined for every pair $q \in Q$, $x \in \Sigma$, the DFA is called \emph{total}. Usually such DFAs are called complete, but we prefer the term ``total'' both to avoid the clash of terminology and because of the obvious connection with total functions. Note that every total DFA is complete, but a complete DFA does not have to be total if it is not strongly connected. A strongly connected DFA is complete if and only if it is total. 
Similarly to NFAs, for a state $q \in Q$ and a word $w \in \Sigma$ of a total DFA, we denote by $q \cdot w$ the state $\delta(q, w)$.

\subparagraph*{Synchronising total DFAs.}
A total DFA is called \emph{synchronising} if there exists a word mapping every state to the same state. Synchronisation of total DFAs is an actively investigated topic in automata theory~\cite{Volkov2008,Kari2019,Volkov2022}, transformation semigroups~\cite{Salomaa2002,Salomaa2003}, combinatorial matrix theory~\cite{Gerencser2018,Wu2023}, coding theory~\cite{Berstel2010,Berlinkov2021} and group theory~\cite{Araujo2017}. It is the subject of the \v{C}ern\'{y} conjecture, one of the oldest open problems in combinatorial automata theory, see~\cite{Kari2019,Volkov2022} for a survey. It is known that if the \v{C}ern\'{y} conjecture is true for strongly connected total DFAs, then it is true for all total DFAs~\cite[Section 3.3]{Volkov2022}.

\subparagraph*{Unambiguous NFAs and diamonds.}
A \emph{diamond} in an NFA $\Aa = (Q, \Sigma, \Delta)$ is a pair of states $p, q \in Q$ and a word $w \in \Sigma^*$ such that $w$ labels two different paths from $p$ to $q$ (see e.g.~\cite{Baier2023}). More formally, states $p, q$ and a word $w$ form a diamond if $w$ can be represented as $w = w_1w_2$
such that there exist states $t_1 \ne t_2$ with $p \cdot w_1 \ni t_1, t_1 \cdot w_2 \ni q$ and $p \cdot w_1 \ni t_2, t_2 \cdot w_2 \ni q$, see \Cref{fig:diamond} (right) for an illustration. An NFA is called \emph{unambiguous} (or sometimes \emph{diamond-free}) if it does not contain any diamonds. Clearly, every DFA is an unambiguous~NFA.

Deciding if a given NFA is unambiguous is \NL-complete, even if the NFA is promised to be strongly connected (see subsection ``Applications to variable-length codes'' of Section~2 in~\cite{Kiefer2025}). In~\cite{Drabik2025}, the fine-grained time complexity of deciding if an NFA is unambiguous is studied, but the notion of NFAs there includes initial and final states and thus slightly differs from ours. However, for strongly connected NFAs, it coincides with ours if one takes an arbitrary state to be the only initial and final state.

Let $\MM$ be a set of nonnegative matrices, and $\Aa$ be the corresponding NFA. The fact that $\Aa$ is unambiguous has the following interpretation in terms of matrices: for any two matrices $M_1, M_2$ from the monoid generated by $\MM$, in the process of computing the product of $M_1$ and $M_2$ no two strictly positive entries are summed up. In particular, if $\MM$ is an irreducible set of matrices whose entries are only zero and one, $\Aa$ is unambiguous if and only if $\MM$ generates a finite monoid. Finite monoids of nonnegative matrices and corresponding automata are studied in~\cite{Weber1991,Weber1991Mat, Kiefer2021,Kiefer2025}.

\subparagraph*{Completeness.}
An NFA $\Aa = (Q, \Sigma, \Delta)$ is \emph{complete} if for every word $w \in \Sigma^*$, there exists a state $q \in Q$ such that $q \cdot w \ne \emptyset$. A word violating this property is called \emph{mortal}, and an NFA admitting a mortal word is called \emph{incomplete}.
The problem of deciding if an NFA is complete is \PSPACE-complete, even for binary strongly connected NFAs~\cite{Kao2009}. However, for unambiguous NFAs it becomes solvable in polynomial time, more precisely in \DET~\cite{Kiefer2021}. To the best of our knowledge, no  lower bounds have been previously shown for the complexity of deciding if an unambiguous NFA is complete. 

Besides the fact that completeness of NFAs is a decidable special case of matrix mortality (as discussed in the previous section), it is also important in language theory and symbolic dynamics, where it can be seen as the fact that every word belongs to the set of factors of a language (called factor language universality)~\cite{Mika2021,Rampersad2012,Kao2009,Gawrychowski2020}. This is especially important in the theory of codes, as we discuss in \Cref{sec:codes}.

\subparagraph*{A remark about complexity classes.}
In this paper, we refer to the classes 
\[\ACZ \subseteq \LL \subseteq \NL = \coNL \subseteq \DET \subseteq \NCT  \subseteq \PTIME.\]
We refer to~\cite{Goldreich2008,Santha1998} for formal definitions, and only provide a brief description and intuition. \ACZ is the class of problems solvable by constant-depth and polynomial-size Boolean circuits. It is a class that lies significantly lower than the other classes that we consider, and in the context of this paper it is synonymous to properties that are extremely easy to decide. In particular, all reductions in this paper are \ACZ reductions. \LL and \NL are the classical classes of problems solvable in deterministic and nondeterministic logarithmic space. We will often use the fact that \NL is closed under complementation, that is, $\NL = \coNL$~\cite[Section~8.6]{Sipser2013}. \NCT is the class of problems solvable by $\OO((\log n)^2)$-depth polynomial-size bounded fan-in Boolean circuits. Finally, \DET is the class of problems that are reducible to verifying the value of the determinant of a matrix with integer entries~\cite{Santha1998}. Intuitively, this class represents problems that are solvable by ``elementary'' linear-algebraic algorithms.

\subsection*{Our contributions} 

In~\Cref{sec:constr-reach} we introduce a variant of the $(s, t)$-reachability problem in digraphs which is crucial for all our constructions and might be of independent interest. Its key property is that it gives a lot of control over possible $(s, t)$-paths. Using this variant, we then prove in~\Cref{sec:nl-sync} that deciding if a strongly connected total DFA is synchronising is \NL-complete. As discussed above, previously \NL-completeness was only known without the promise of strong connectivity, and this promise requires a significantly more advanced approach.

We then further extend our approach in~\Cref{sec:comp-and-unamb} by showing two complementary statements: given a strongly connected NFA, it is \NL-hard to decide if it is complete even if it is unambiguous (\Cref{thm:nl-complete}), and it is \NL-hard to decide if it is unambiguous even if it is complete (\Cref{thm:nl-unamb}). As mentioned above, no lower bounds on deciding completeness of an unambiguous NFA were known, and deciding unambiguity of strongly connected NFAs was known to be \NL-complete only without the promise that the NFA is complete.

In~\Cref{sec:2-im-b}, we observe that the NFAs in the \NL-hardness reductions from~\Cref{sec:comp-and-unamb} belong to the class that we call $2$-image-bounded. In matrix terms, this means that every matrix in the monoid contains at most two strictly positive entries in each row. It can also be seen as a natural generalisation of DFAs, since for DFAs the matrices have at most one strictly positive entry per row.
Given such a natural definition, it is quite surprising that we were not able to find any results about this class in the literature. 

In \Cref{sec:codes}, we discuss the consequences of our results in the theory of codes, a natural source of problems for strongly connected automata. Finally, \Cref{sec:conclusions} presents conclusions and open problems.

%% file: sec-constr-reach.tex
Our technique relies on the following lemma. Its statement is quite technical, but the main idea is that, for the length of a shortest $(s, t)$-path in a digraph, deciding between the options that this length is $n - 1$ or $n$ is \NL-hard. This gives us a lot more control over reachability properties of digraphs compared to the classical $(s, t)$-reachability problem.

\begin{lemma} \label{lem:st-reach-ass}
The following problem is \NL-complete: given a digraph $G$, two of its vertices~$s, t$ and a number~$n$ with the promises that
\begin{itemize}
    \item $G$ is acyclic, 

    \item $t$ has outdegree zero,

    \item all other vertices have outdegree at most two, 

    \item every $(s, t)$-path in $G$ has length $n - 1$ or $n$,

    \item and for every vertex $v$ there is a $(v, t)$-path in $G$,
\end{itemize} 
decide if there exists an $(s, t)$-path of length~$n - 1$ in~$G$. 
\end{lemma}

When using this lemma in the remainder of the paper, we will assume that 
the outdegree of each vertex except for $t$ is exactly two, which can be achieved by simply duplicating edges. Formally, this results in a directed multigraph (where edges form a multiset instead of a set), which will not cause any issues in the paper.

\begin{proof}
We provide an \ACZ reduction from the $(s,t)$-reachability problem.
Let $G = (V, E)$ and $s, t \in V$ be the input of this problem. The idea is to construct a digraph $G'$ with $|V| + 1$ layers of states, each containing a copy of each state from $V$. The edges in $G'$ are defined in such a way that, for layer $i$, taking an edge $(u, v)$ in $G$ corresponds to going from the copy of~$u$ in layer $i$ to the copy of $v$ in layer $i + 1$ if $v \ne t$, and to the copy of $v$ in layer $i + 2$ if $v = t$. In the last layer, all vertices are merged, resulting in just one vertex, and this vertex can be reached from the copy of $s$ in the first layer by a path of length $|V| - 1$ if and only if there is an $(s, t)$-path in $G$. By taking $n = |V|$ and renaming $G'$, $s', t'$ as $G$, $s, t$ respectively, we get that the problem in the statement of the lemma is \NL-complete.

We now provide the details of this construction.
We can assume that the outdegree of each vertex in $G$ is at most two. Indeed, if the outdegree of a vertex $u$ is larger than two, we can replace all edges outgoing from it with a directed binary tree whose root is $u$ and whose leaves are the vertices $v \in V$ such that $(u, v) \in E$.  
Clearly, this operation does not affect the existence of a path from $s$ to~$t$. 
To obtain a reduction to the problem in the statement, we take $n = |V|$ and create a new digraph $G' = (V', E')$ whose state set $V'$ consists of $n$ copies of $V$ and a new state $t'$. The set $E'$ of edges consists of the following four types:
\begin{enumerate}[(i)]
    \item for each $1 \le i \le n - 1$ and each edge $(u, v) \in E$ with $u,v \ne t$, create an edge going from the $i$th copy of $u$ to the $(i + 1)$st copy of $v$;

    \item for each vertex $u \in V$, create an edge going from the $n$th copy of $u$ to $t'$;

    \item for each $1 \le i \le n - 1$ and for each vertex $v \in V$ of outdegree zero (including in particular $t$), create an edge going from the $i$th copy of $v$ to the $(i + 1)$st copy of $v$;

    \item for each $1 \le i \le n - 1$ and each edge $(u, t) \in E$, create an edge going from the $i$th copy of $u$ to the $(i + 2)$nd copy of $t$ ($t'$ is considered the $(n+1)$st copy of~$t$).
\end{enumerate}

For each vertex $v \in V$, denote the $i$th copy of $v$ as $v^{(i)}$. Take $s' = s^{(1)}$ and $t' = t^{(n+1)}$. See \Cref{fig:st-reachability} for an~example of the described construction.

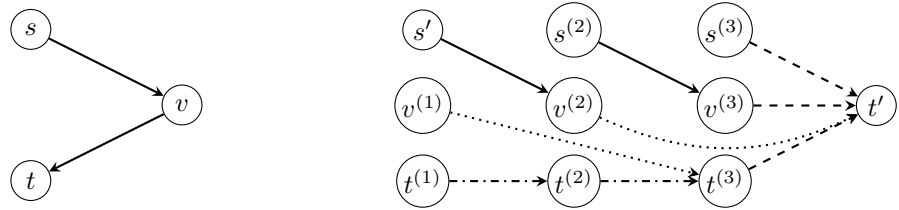
\begin{figure}[ht]\centering
\begin{subfigure}[c]{0.45\textwidth} \centering
\begin{tikzpicture} [node distance = 2cm]
\tikzset{every state/.style={inner sep=1pt,minimum size=1.5em}}

\node [state] at (0, 1) (s) {$s$};
\node [state] at (2, 0) (v) {$v$};
\node [state] at (0, -1) (t) {$t$};
\path [-stealth, thick]

(s) edge [] node[above] {} (v)
(v) edge [] node[above] {} (t)

;
\end{tikzpicture}
\end{subfigure}
\begin{subfigure}[c]{0.45\textwidth} \centering
\begin{tikzpicture} [node distance = 2cm]
\tikzset{every state/.style={inner sep=1pt,minimum size=1.5em}}

\node [state] at (0, 1) (s1) {$s'$};
\node [state] at (0, 0) (v1) {$v^{(1)}$};
\node [state] at (0, -1) (t1) {$t^{(1)}$};

\node [state] at (2, 1) (s2) {$s^{(2)}$};
\node [state] at (2, 0) (v2) {$v^{(2)}$};
\node [state] at (2, -1) (t2) {$t^{(2)}$};

\node [state] at (4, 1) (s3) {$s^{(3)}$};
\node [state] at (4, 0) (v3) {$v^{(3)}$};
\node [state] at (4, -1) (t3) {$t^{(3)}$};

\node [state] at (6, 0) (t) {$t'$};

\path [-stealth, thick]

(s1) edge [] node[above] {} (v2)
(v1) edge [dotted] node[above] {} (t3)

(s2) edge [] node[above] {} (v3)
(v2) edge [dotted, bend right=25] node[above] {} (t)

(t1) edge [dashdotted] node[above] {} (t2)

(t2) edge [dashdotted] node[above] {} (t3)

(s3) edge [dashed] node[above] {} (t)
(v3) edge [dashed] node[above] {} (t)
(t3) edge [dashed] node[above] {} (t)

;
\end{tikzpicture}
\end{subfigure}

\caption{The original digraph $G$ (left) and the digraph $G'$ obtained in the reduction in the proof of \Cref{lem:st-reach-ass} (right). We have $n=3$. Solid edges correspond to making one step forward (type (i)), dashed edges to making the final step (type (ii)), dashdotted edges to progressing from a vertex of outdegree zero (type (iii)), and dotted edges to making two steps forward (type (iv)).}\label{fig:st-reachability}
\end{figure}

Suppose there is an $(s,t)$-path in~$G$, say $s_1 \cdots s_k t$ with $s_1 = s$.
We can assume that~$k < n$.
Then there is an $(s',t')$-path in~$G'$, namely $s_1^{(1)} \cdots s_{k}^{(k)} t^{(k+2)} \cdots t^{(n+1)}$.
Note that $(s_{k}^{(k)}, t^{(k+2)})$ is an edge of type (iv).

Conversely, suppose that $G'$ has an $(s',t')$-path, say~$\rho$, of length at most $n-1$.
Then~$\rho$ uses an edge of type (iv), say $(s_{k}^{(k)},t^{(k+2)})$ for some $(s_k,t) \in E$.
Since $\rho$ cannot use another edge of type~(iv), the length of~$\rho$ is $n-1$.
We are going to show that before the edge $(s_{k}^{(k)},t^{(k+2)})$, the path~$\rho$ can only use edges of type~(i).
Indeed, the source of any edge of type (iii) is of the form $v^{(i)}$, where $v$ has outdegree zero in~$G$; but in~$G'$ the node $v^{(i)}$ can only reach~$t'$ or nodes of the form $v^{(j)}$ with $j \ge i$; in particular, $v^{(i)}$ cannot reach $s_{k}^{(k)}$, as $s_k$~does not have outdegree zero.
Thus, $\rho$ has a prefix $s_1^{(1)} \cdots s_k^{(k)}$ (where $s' = s_1^{(1)}$) that consists entirely of edges of type~(i).
Hence, by construction, $s_1 \cdots s_k t$ is an $(s,t)$-path in~$G$.
\end{proof}

%% file: sec-nl-sync.tex
\comment{The \emph{rank} of a total DFA $\Aa = (Q, \Sigma, \delta)$ is the minimum size of the set $\{q \cdot w \mid q \in Q\}$ among all words $w \in \Sigma^*$. Clearly, the rank of a total DFA is one if and only if it is synchronising.} The first main result of this paper is as follows.

\begin{theorem}\label{thm:nl-sync}
    Deciding if a binary strongly connected total DFA\comment{of rank at most two} is synchronising is \NL-complete.
\end{theorem}

As mentioned in~\Cref{subs:two-examples}, containment in \NL is well known and easy to show.
We thus need to prove \NL-hardness.
We do so by reducing from the problem in the statement of \Cref{lem:st-reach-ass}. Let $G = (V, E)$, $s, t \in V$ and a natural number $n$ be its input. As remarked after the statement of \Cref{lem:st-reach-ass}, we can assume that all vertices in~$G$ except for~$t$ have outdegree two. We construct a binary total DFA $\Aa = (Q, \Sigma, \delta)$ with $\Sigma = \{a, b\}$. 

\subparagraph*{The idea.} A high-level outline of the construction is as follows, see \Cref{fig:nl-sync} for an illustration. Our construction consists of two almost identical halves, which we refer to as the top and the bottom part following \Cref{fig:nl-sync}. For $i \in \{1, 2\}$, part $i$ consists of a copy of $G$ on the state set~$V_{(i)}$, a tree~$T_{(i)}$ and two states $r_{(i)}, q_{(i)}$. Moreover, the bottom part contains a timer gadget creating a path of length~$n$. The only role of the trees $T_{(i)}$ is to make sure that every state in $V_{(i)}$ can be reached from~$r_{(i)}$, so that the resulting DFA is strongly connected.

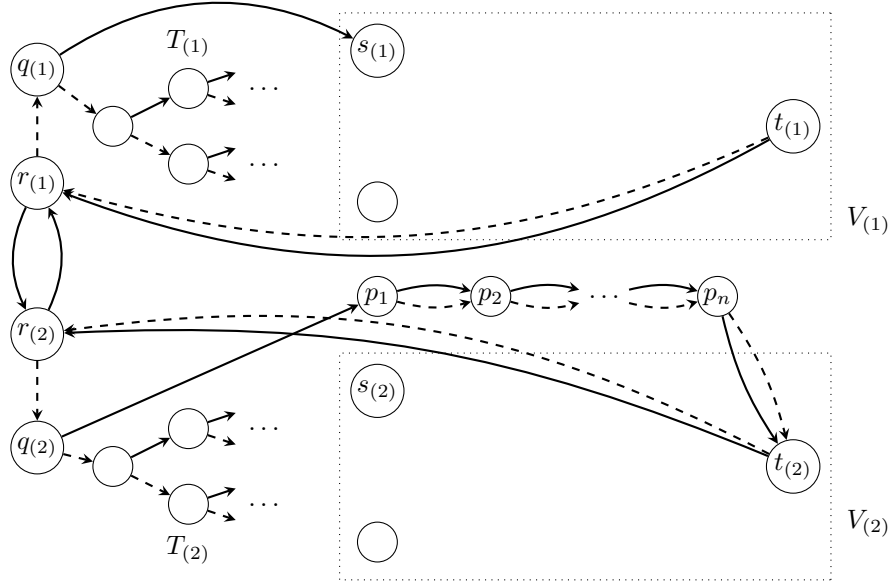
\begin{figure}[ht]\centering
 \centering
\begin{tikzpicture} [node distance = 2cm]
\tikzset{every state/.style={inner sep=1pt,minimum size=1.5em}}

\node [state] at (0, 2.5) (u1) {$q_{(1)}$};
\node [state] at (0, 1) (v1) {$r_{(1)}$};

\node [state] at (1, 1.75) (tree1) {};
\node [state] at (2, 2.25) (tree10) {};
\node [state] at (2, 1.25) (tree11) {};
\node [] at (2.75, 2.5) (tree100) {};
\node [] at (2.75, 2) (tree101) {};
\node [] at (2.75, 1.5) (tree110) {};
\node [] at (2.75, 1) (tree111) {};

\node [] at (2, 2.85) () {$T_{(1)}$};

\node [] at (3, 2.25) () {$\ldots$};
\node [] at (3, 1.25) () {$\ldots$};

\node [state] at (4.5, 2.75) (s1) {$s_{(1)}$};
\node [state] at (4.5, 0.75) (t1) {};
\node [state] at (10, 1.75) (t1end) {$t_{(1)}$};

\draw [draw=black,dotted] (10.5,3.25) rectangle (4,0.25);

\node [] at (11, 0.5) (t1) {$V_{(1)}$};


\node [state] at (0, -1) (v2) {$r_{(2)}$};
\node [state] at (0, -2.5) (u2) {$q_{(2)}$};

\node [state] at (1, -2.75) (ctree1) {};
\node [state] at (2, -3.25) (ctree10) {};
\node [state] at (2, -2.25) (ctree11) {};
\node [] at (2.75, -3.5) (ctree100) {};
\node [] at (2.75, -3) (ctree101) {};
\node [] at (2.75, -2.5) (ctree110) {};
\node [] at (2.75, -2) (ctree111) {};

\node [] at (2, -3.85) () {$T_{(2)}$};

\node [] at (3, -3.25) () {$\ldots$};
\node [] at (3, -2.25) () {$\ldots$};

\node [state] at (4.5, -3.75) (cs1) {};
\node [state] at (4.5, -1.75) (ct1) {$s_{(2)}$};
\node [state] at (10, -2.75) (ct1end) {$t_{(2)}$};

\draw [draw=black,dotted] (10.5,-4.25) rectangle (4,-1.25);
\node [] at (11, -3.5) (t1) {$V_{(2)}$};


\node [state] at (4.5, -0.5) (time1) {$p_1$};
\node [state] at (6, -0.5) (time2) {$p_2$};
\node [] at (7.5, -0.5) (time3) {$\ldots$};
\node [state] at (9, -0.5) (time4) {$p_n$};

\path [-stealth, thick]

(v1) edge [dashed] node[above] {} (u1)
(v2) edge [dashed] node[above] {} (u2)

(v1) edge [bend right=25] node[above] {} (v2)
(v2) edge [bend right=25] node[above] {} (v1)

(u1) edge [bend left=30] node[above] {} (s1)

(t1end) edge [dashed,bend left=20] node[above] {} (v1)
(t1end) edge [bend left=25] node[above] {} (v1)

(u1) edge [dashed] node[above] {} (tree1)
(tree1) edge [] node[above] {} (tree10)
(tree1) edge [dashed] node[above] {} (tree11)

(tree10) edge [] node[above] {} (tree100)
(tree10) edge [dashed] node[above] {} (tree101)

(tree11) edge [] node[above] {} (tree110)
(tree11) edge [dashed] node[above] {} (tree111)


(u2) edge [dashed] node[above] {} (ctree1)
(ctree1) edge [dashed] node[above] {} (ctree10)
(ctree1) edge [] node[above] {} (ctree11)

(ctree10) edge [dashed] node[above] {} (ctree100)
(ctree10) edge [] node[above] {} (ctree101)

(ctree11) edge [dashed] node[above] {} (ctree110)
(ctree11) edge [] node[above] {} (ctree111)

(ct1end) edge [dashed,bend right=17.5] node[above] {} (v2)
(ct1end) edge [bend right=12.5] node[above] {} (v2)


(u2) edge [] node[above] {} (time1)
(time1) edge [dashed,bend right=15] node[above] {} (time2)
(time1) edge [bend left=15] node[above] {} (time2)

(time2) edge [dashed,bend right=15] node[above] {} (time3)
(time2) edge [bend left=15] node[above] {} (time3)

(time3) edge [dashed,bend right=15] node[above] {} (time4)
(time3) edge [bend left=15] node[above] {} (time4)

(time4) edge [bend right=10] node[above] {} (ct1end)
(time4) edge [dashed, bend left=10] node[above] {} (ct1end)

;
\end{tikzpicture}
\caption{The construction in the proof of \Cref{thm:nl-sync}. Transitions by letter $a$ are depicted by solid arrows, and by letter $b$ by dashed arrows.
}\label{fig:nl-sync}
\end{figure}

The construction ensures that there is a word that maps each state to either $r_{(1)}$ or $r_{(2)}$. The action of every word is symmetric on $r_{(1)}$ and $r_{(2)}$ in the corresponding parts, except  for $ba$, which maps $r_{(1)}$ to~$s_{(1)}$ and $r_{(2)}$ to the first state, $p_1$, of the timer gadget.
The only way to break the symmetry between the two parts is to then apply a word that makes $s_{(1)}$ reach~$t_{(1)}$ before or after $p_1$ reaches~$t_{(2)}$.
By our assumptions, that is possible if and only if~$G$ has an $(s,t)$-path of length $n-1$.
If such a path exists, the corresponding word maps $s_{(1)}$ to~$t_{(1)}$ and $p_1$ to~$p_n$.
Subsequent application of $a a$ maps both $t_{(1)}$ and $p_n$ to~$r_{(2)}$.

\subparagraph*{Formal construction.} Formally, the set $Q$ of states consists of two copies $\{r_{(1)}, q_{(1)}\} \cup V_{(1)} \cup T_{(1)}$ and $\{r_{(2)}, q_{(2)}\} \cup V_{(2)} \cup T_{(2)}$ together with a set $\{p_1, \ldots, p_n\}$. For each state $h_{(1)}$ of the first copy we refer to $h_{(2)}$ as its symmetric state, and vice versa.
Here, $V_{(1)}, V_{(2)}$ are two separate copies of $V$ and $T_{(1)}, T_{(2)}$ are two disjoint sets of auxiliary fresh states that we will define later. The transition function~$\delta$ is defined as follows (recall the notation~$q \cdot x = \delta(q, x)$). 

We take $q_{(1)} \cdot a = s_{(1)}$, $q_{(2)} \cdot a = p_1$, $r_{(1)} \cdot a = r_{(2)}$ and $r_{(2)} \cdot a = r_{(1)}$. For each $1 \le i \le n - 1$ and $x \in \{a, b\}$ we define $p_i \cdot x = p_{i + 1}$, and for each $x \in \{a, b\}$ we define $p_n \cdot x = t_{(2)}$.
All other transitions are defined symmetrically for states in their corresponding copies, meaning that if for some $h_{(1)}$ we have $h_{(1)} \cdot x = g_{(1)}$ for $x \in \{a, b\}$, then $h_{(2)} \cdot x = g_{(2)}$. We thus define~$\delta$ only for the first copy.
 We take $r_{(1)} \cdot b = q_{(1)}$. Letter $b$ maps $q_{(1)}$ to the root of a binary tree whose non-leaves are the states in~$T_{(1)}$. This tree is defined in such a way that its leaves are all states in~$V_{(1)}$, each inner vertex has outdegree two, and all paths from the root to the leaves have the same length. The action of $a, b$ on the vertices of this tree is defined arbitrarily in accordance to the edges. By construction, for some $\ell$, every word of length $\ell$ maps the root to a state in~$V_{(1)}$.
 For the states in $V_{(1)}\setminus \{t_{(1)}\}$, the action of $a$ and~$b$ is defined arbitrarily in accordance to the edges in $G$, using the fact that the outdegree of each vertex in $V \setminus \{t\}$ is two. Finally, we define $t_{(1)} \cdot x = r_{(1)}$ for $x \in \{a, b\}$. See \Cref{fig:nl-sync} for an illustration.

\begin{claim} $\Aa$ is strongly connected.
\end{claim}
\begin{proof}
    
By construction, every state in $V_{(1)}$ can be reached from $r_{(1)}$ by following a path in the tree $T_{(1)}$. Similarly, we can reach $t_{(1)}$, and thus $r_{(1)}$, from every state in $V_{(1)}$ by the assumptions of \Cref{lem:st-reach-ass}. The same holds symmetrically for $r_{(2)}$. Since $r_{(1)}$ and $r_{(2)}$ can be reached from each other, we get that~$\Aa$ is strongly connected.
\end{proof}

\begin{claim}
    $\Aa$ is synchronising if and only if there exists an $(s, t)$-path of length $n - 1$ in $G$.
\end{claim}
\begin{proof}
Observe that a repeated application of $a$ eventually maps every state to $r_{(1)}$ or $r_{(2)}$. It thus remains to prove that $r_{(1)}, r_{(2)}$ can be mapped to the same state if and only if there exists an $(s, t)$-path of length $n - 1$ in $G$.

Assume first that such a path exists, and let $w$ be the label of such a path according to our definition of $\delta$ for states in $V_{(1)}$.
Then the word $bawaa$ maps both $r_{(1)}$ and $r_{(2)}$ to~$r_{(2)}$.

Conversely, suppose that there is a word $w$ mapping $r_{(1)}$ and $r_{(2)}$ to the same state. Without loss of generality, this word has to begin with $ba$, since it is the only way to break the symmetry. Let $w = baw_1w_2$, where $|w_1| = n$ (if $w$ is shorter, it cannot map $r_{(1)}$ and $r_{(2)}$ to the same state). Note that $r_{(2)} \cdot baw_1 = t_{(2)}$. We can thus assume that $r_{(1)} \cdot baw_1 \ne t_{(1)}$, since otherwise the symmetry is restored. Hence,  $s_{(1)} \cdot w_1 \ne t_{(1)}$, which is only possible if there exists an $(s, t)$-path of length different to $n$ in $G$. By assumptions on $G$, such a path must have length $n - 1$.
\end{proof}

\begin{remark} \normalfont
The total DFA in the construction has rank at most two, as can be seen in the proof of the previous claim. Here, the \emph{rank} of a total DFA is the minimum size of the set $\{q \cdot w \mid q \in Q\}$ for all words $w$. Hence, we get that \NL-hardness in \Cref{thm:nl-sync} holds true even with the additional promise that the rank of the input total DFA is at most two.
\end{remark}

%% file: sec-nl-mortality.tex
We now prove the second main result of the paper.

\begin{theorem}\label{thm:nl-complete}
    Deciding if a binary strongly connected unambiguous NFA is complete is \NL-hard.
\end{theorem}

To prove the theorem, we reduce from the constrained $(s, t)$-reachability problem described in \Cref{lem:st-reach-ass}. Let $G = (V, E)$, $s, t \in V$ and a natural number $n$ be its input. We construct an NFA $\Aa = (Q, \Sigma, \Delta)$ with $\Sigma = \{a, b, c\}$. To further reduce it to the case of the binary alphabet, one can use a standard technique: duplicate one of the letters, and replace for each state all outgoing transitions with a complete binary directed tree of depth two. In other words, replace letters $a, b, c, c'$ (where $c'$ acts in the same way as $c$) with words $xx, xy, yx, yy$ respectively. Clearly, this operation preserves completeness, unambiguity and strong connectivity.

\subparagraph*{The idea.} A high-level outline of the construction is as follows, see \Cref{fig:nl-mortality} for an illustration. The construction includes two copies of $G$ identically labelled by $a$ and $b$ so that all the transitions are deterministic. It also includes a complete binary tree rooted in state $f$, whose edges are also labelled by $a$ and $b$ so that all the transitions are deterministic. This tree is used to guarantee that every state of both copies of $G$ can be reached from $f$, so that the NFA is strongly connected.

So far, the described construction is symmetric with respect to the two copies of $G$. In the remainder of the construction this symmetry is broken so that state $f$ can be killed if and only if there is a ``shortcut'' (that is, an $(s,t)$-path of length $n - 1$) in~$G$. 
To achieve that, we add a timer gadget of length $n$ similar to the construction in the proof of~\Cref{thm:nl-sync} (see~\Cref{fig:nl-mortality}). Letter $c$ maps $f$ to the set consisting of the state corresponding to $s$ in the second copy of $G$ and the first state of the timer gadget. If there is a ``shortcut'' in $G$, $f$ can be killed by applying $c$, then the label of the shortcut, and then $c$ again ($c$ kills all the states of the timer gadget except the last one).
To guarantee that $f$ cannot be killed without a ``shortcut'', the last state of the timer gadget has self-loops for $a$ and $b$, and letter $c$ maps all states from the second copy of $G$ and the last state of the timer gadget back to $f$. The first copy of $G$ is used to make sure that $f$ cannot be killed in any other way. 
Moreover, our construction ensures that if $f$ can be killed, there is also a word that kills all states.

\begin{figure}[ht]\centering
 \centering
\begin{tikzpicture} [node distance = 2cm]
\tikzset{every state/.style={inner sep=1pt,minimum size=1.5em}}


\node [state] at (4.5, 2.75) (s1) {$s_{(1)}$};
\node [state] at (4.5, 0.75) (t1) {};
\node [state] at (10.5, 1.75) (t1end) {$t_{(1)}$};

\node [state] at (7.25, 2.25) (c1) {};

\draw [draw=black] (11,2.25) rectangle (10,1.25);

\draw [draw=black,dotted] (11,3.25) rectangle (4,0.25);

\node [] at (11.5, 0.5) (t1) {$V_{(1)}$};


\node [state] at (-0.5, 0) (ctree1) {$f$};

\node [state] at (1, -0.75) (ctree10) {};
\node [state] at (1, +0.75) (ctree11) {};
\node [] at (1.75, -1) (ctree100) {};
\node [] at (1.75, -0.5) (ctree101) {};
\node [] at (1.75, 0.5) (ctree110) {};
\node [] at (1.75, 1) (ctree111) {};

\node [] at (2.25, 1) (preleaf) {};
\node [state] at (3, 0.75) (leaf) {};

\node [] at (1.5, 1.5) () {$T$};

\draw [draw=black] (3.5,1.25) rectangle (0.5,-1.25);

\node [] at (2, -0.75) () {$\ldots$};
\node [] at (2, 0.75) () {$\ldots$};

\node [state] at (4.5, -2.75) (cs1) {};
\node [state] at (4.5, -0.75) (ct1) {$s_{(2)}$};
\node [state,dashdotted] at (10.5, -1.75) (ct1end) {$t_{(2)}$};

\node [state] at (7.25, -1.25) (c2) {};

\draw [draw=black,dotted] (11,-3.25) rectangle (4,-0.25);
\node [] at (11.5, -3) (t1) {$V_{(2)}$};

\draw [draw=black] (10,-3.25) rectangle (4,-0.25);

\node [state] at (4.5, -4) (time1) {$p_1$};
\node [state] at (6, -4) (time2) {$p_2$};
\node [] at (7.5, -4) (time3) {$\ldots$};
\node [state] at (9, -4) (time4) {$p_n$};
\node [state] at (10.5, -4) (time5) {$p_{n+1}$};

\draw [draw=black] (11,-4.5) rectangle (10,-3.5);

\path [-stealth, thick]


(t1end) edge [loop right,dashed] node[above] {} (t1end)
(t1end) edge [loop right,min distance=10mm,in=-30,out=30] node[above] {} (t1end)

(time5) edge [loop right,dashed] node[above] {} (time5)
(time5) edge [loop right,min distance=10mm,in=-30,out=30] node[above] {} (time5)


(ctree1) edge [dashed] node[above] {} (ctree10)
(ctree1) edge [] node[above] {} (ctree11)

(ctree10) edge [dashed] node[above] {} (ctree100)
(ctree10) edge [] node[above] {} (ctree101)

(ctree11) edge [dashed] node[above] {} (ctree110)
(ctree11) edge [] node[above] {} (ctree111)


(time1) edge [dashed,bend right=15] node[above] {} (time2)
(time1) edge [bend left=15] node[above] {} (time2)

(time2) edge [dashed,bend right=15] node[above] {} (time3)
(time2) edge [bend left=15] node[above] {} (time3)

(time3) edge [dashed,bend right=15] node[above] {} (time4)
(time3) edge [bend left=15] node[above] {} (time4)

(time4) edge [dashed, bend right=15] node[above] {} (time5)
(time4) edge [bend left=15] node[above] {} (time5)

(ctree1) edge [dotted] node[above] {} (ct1)
(ctree1) edge [dotted] node[above] {} (time1)

(preleaf) edge [dashed] node[above] {} (leaf)

(leaf) edge [bend left=5] node[] {} (c1)
(leaf) edge [dashed, bend right=5] node[above] {} (c1)

(leaf) edge [bend left=5] node[above] {} (c2)
(leaf) edge [dashed, bend right=5] node[above] {} (c2)
;
\end{tikzpicture}
\caption{The construction in the proof of \Cref{thm:nl-complete}. Transitions by $a$ are depicted by solid arrows, by $b$ by dashed arrows, and by $c$ by dotted arrows. All states in solid rectangles are mapped by $c$ to~$\{f\}$. The dashdotted state $t_{(2)}$ is deleted together with all incoming transitions.
}\label{fig:nl-mortality}
\end{figure}
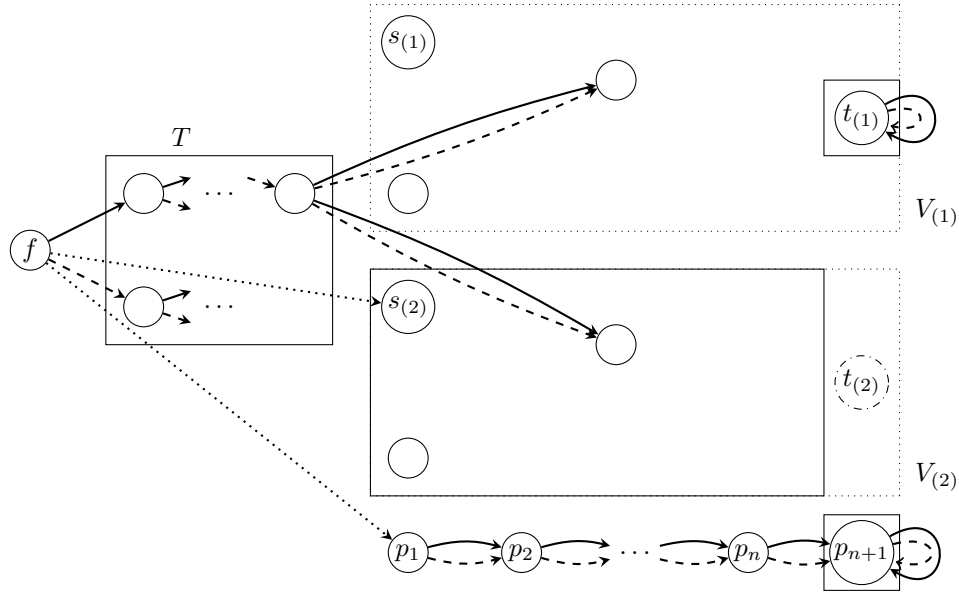

\subparagraph*{Formal construction.} Let us now describe the details of the construction. The set $Q$ consists~of
\begin{itemize}
    \item two copies of $V$, denoted $V_{(1)}$ and $V_{(2)}$; the states corresponding to $s$ and $t$ in the $i$th copy, $i \in \{1, 2\}$, are denoted $s_{(i)}$ and $t_{(i)}$ respectively;

    \item a set $T$ of states of the tree that we will define later, with a root $f \in T$;

    \item a set $P = \{p_1, \ldots, p_n, p_{n + 1}\}$ of states of the timer gadget.
\end{itemize}

The transitions of $\Aa$ are defined as follows. The edges of both copies of $G$ are labelled by $a$ and~$b$ so that all obtained transitions are deterministic and the labellings are the same for both copies (recall that the outdegree of each vertex in $G$ except for $t$ is two). State~$t_{(1)}$ is mapped by both $a$ and $b$ to itself.
State $t_{(2)}$ is deleted together with all incoming transitions.

For states from the set $T$, the transitions by $a$ and $b$ are defined so that the result is a complete binary tree, all the transitions are deterministic, and the number of leaves is $|V|$ (if~$|V|$ is not a power of two, merge some of the leaves). In particular, all paths from the root~$f$ to the leaves have the same length. The number of states in~$T$ is chosen accordingly. We slightly abuse the notation and refer by~$T$ to both the tree and its set of states. 
We assign to each leaf state of $T$ a separate vertex in $V$, and both $a$ and $b$ map each leaf state to a set consisting of the copies of the corresponding vertex in $V_{(1)}$ and $V_{(2)}$.
For $1 \le i \le n$, $p_i$ is mapped by both $a$ and $b$ to $p_{i + 1}$, and $p_{n + 1}$ is mapped by both $a$ and $b$ to itself.

It remains to define the action of letter $c$. Each state in $T \setminus \{f\}$, $V_{(2)}$ and $\{t_{(1)}, p_{n + 1}\}$ is mapped by $c$ to $f$. We also define $f \cdot c = \{s_{(2)}, p_1\}$. Every  other state is mapped by $c$ to the empty set. The remainder of the proof directly follows from the next three claims.

\begin{claim}
    $\Aa$ is strongly connected.
\end{claim}
\begin{proof}
    Every state in $V_{(1)}$ and $V_{(2)}$ can be reached from $f$ via~$T$, $t_{(1)}$ can be reached from every state of $V_{(1)}$ by the assumption of~\Cref{lem:st-reach-ass}, and every state in $\{t_{(1)}\} \cup V_{(2)} \cup \{p_{n + 1}\}$ is mapped by $c$ back to~$f$.
\end{proof}

\begin{claim}
    $\Aa$ is unambiguous.
\end{claim}
\begin{proof}
The only states $q$ with $|q \cdot x| \ge 2$ for some letter $x$ are $f$ and the leaf states of $T$. A diamond can only start in such a state.

Consider first the case of a leaf state $q$ of $T$. The image of $q$ under $w \in \{a, b\}^*$ consists either of two states that are copies of the same vertex of $G$ (then the application of $c$ kills the copy in $V_{(1)}$), or of only one state $t_{(1)}$ (hence there cannot be a diamond by definition).

Consider now paths starting in $f$ and labelled by a word $w$. To induce a diamond, $w$ must start with $c$. Let $w = cw_1w_2$, where $w_1 \in \{a, b\}^*$ and $w_2$ is an empty word or starts with $c$. If $|w_1| < n$, then the application of $c$ after $w_1$ kills the image of $f$ in $P$, and if $|w_1| \ge n$ then the application of $w_1$ kills the image in $V_{(2)}$. In both cases, there can be no diamond. 
\end{proof}

\begin{claim}\label{claim:iff-nl-compl}
    $\Aa$ is complete if and only if there is no $(s, t)$-path of length $n - 1$ in $G$.
\end{claim}
\begin{proof}
Assume first that there is an $(s, t)$-path of length $n - 1$ in $G$, and let $w \in \{a, b\}^{n - 1}$ be a word labelling the corresponding path in $V_{(1)}$. We claim that $cwccwc$ is a mortal word. Observe first that $cwc$ kills~$f$. Indeed, $f \cdot c = \{s_{(2)}, p_1\}$, and there is at most one path starting in each of $s_{(2)}, p_1$ and labelled by $w$. By construction, $w$ kills $s_{(2)}$, since it maps it to $t_{(2)}$, which is deleted. Moreover, $wc$ kills~$p_1$, since $|w| = n - 1$.
It remains to note that the image of every state $q \ne f$ under $cwc$ is either~$\{f\}$ or the empty set. Indeed, every state $q \ne f$ is mapped by $c$ to $\{f\}$ or the empty set, then $f$ is mapped by $w$ to a subset of $T \cup V_{(1)} \cup V_{(2)}$, and then this subset is mapped by $c$ to $\{f\}$ by construction. Hence, the second application of $cwc$ kills all remaining states.

Now we assume that all $(s, t)$-paths in $G$ are of length $n$. We are going to show that in this case for every word $w \in \Sigma^*$ we have $f \cdot w \ne \emptyset$. This obviously implies that $\Aa$ is complete.

If $w \in \{a, b\}^*$, then $f \cdot w$ is obviously non-empty by construction. 
Assume now that $w = w_1cw_2$, where $w_1 \in \{a, b\}^*$. If $f \cdot w_1 = \{q\}$ for $q \in T$, then $f \cdot w_1 c = \{f\}$, and we can proceed by induction by the number of occurrences of $c$ in $w$. Similarly, if $f$ is mapped by~$w_1$ to a subset of $V_{(1)} \cup V_{(2)}$, then depending on the length of $w_1$ the image $f \cdot w_1$ is either a set of the copies of a vertex from $V$, or $t_{(1)}$. In both cases we get that $f \cdot w_1 c = \{f\}$. 

It remains to consider the case where $w_1$ is the empty word. If $w_2$ does not contain any occurrences of $c$, then $f \cdot cw_2$ contains either a state from $V_{(2)}$ (if $|w_2| \le n - 1$), or $p_{(n + 1)}$ (if $|w_2| \ge n$). Here, we use the fact that no word of length at most $n - 1$ maps $s_{(2)}$ to the deleted state $t_{(2)}$ since there is no $(s, t)$-path of length $n - 1$ in $G$. Finally, if $c$ occurs in $w_2$, take $w_2 = w_3 c w_4$ with $w_3 \in \{a, b\}^*$. Then $f \cdot c w_3 c = \{f\}$ by the same argument as above, which concludes the proof.
\end{proof}




%% file: sec-nl-unamb.tex
Extending the technique of the previous subsection, we now show the following result.

\begin{theorem}\label{thm:nl-unamb}
    Deciding if a binary strongly connected complete NFA is unambiguous is \NL-complete.
\end{theorem}

The fact that deciding unambiguity of an NFA is in \NL is well-known and is easy to see by reducing it to reachability on pairs of states. 
To prove \NL-hardness, we use a similar construction to the one in the proof in \Cref{thm:nl-complete}. We again reduce from the $(s, t)$-reachability problem with constraints described in \Cref{lem:st-reach-ass}. Let $G = (V, E)$, $s, t \in V$ and a natural number $n$ be its input. We construct an NFA $\Aa = (Q, \Sigma, \Delta)$ with $\Sigma = \{a, b, c\}$ and then make the alphabet binary as described in the previous subsection. We only provide the idea of the construction and the proof, since both of them are similar to those from \Cref{thm:nl-complete}. 

\subparagraph*{The idea.} A high-level outline of the construction is as follows, see \Cref{fig:nl-unamb} for an illustration. The construction includes a copy of $G$ labelled by $a$ and $b$ so that all the transitions are deterministic. Recall that we assume that the outdegree of each vertex in $G$ except for $t$ is precisely two. The construction also includes a complete binary tree $T$ rooted in state $f$. All paths from $f$ to the leaves of $T$ have the same length. The edges of the tree are also labelled by $a$ and $b$ so that all the transitions are deterministic. We denote the states of the copy of~$G$ by~$V'$, and the copies of $s$ and $t$ by $s'$ and $t'$ respectively. The tree $T$ is used to guarantee that every state in $V'$ can be reached from $f$. For $x \in \{a, b\}$, we define $t' \cdot x = \{t'\}$. We also introduce again the timer gadget, now with states $p_1, \ldots, p_n$. For $x \in \{a, b\}$, we define $p_i \cdot x = \{p_{i + 1}\}$ for $1 \le i \le n - 1$, and $p_n \cdot x = \emptyset$.

Letter $c$ maps $f$ to $\{s', p_1\}$, and every state from $T \cup \{t', p_1, \ldots, p_n\}$ to $\{f\}$. Every other state is mapped by $c$ to the empty set.

By construction, a diamond in $\Aa$ can only start in $f$ and be labelled by a word $cwc$ for some $w \in \{a, b\}^*$.
If there is an $(s, t)$-path of length $n - 1$ in $G$, and this path is labelled by $w$ in $\Aa$, then $cwc$ labels two different paths from $f$ to itself. Otherwise, observe that for every $w \in \{a, b\}^*$ the set $f \cdot cw$ consists either of a single state from $V' \setminus \{t'\}$ and a state from the timer gadget $\{p_1, \ldots, p_n\}$ (if $|w| \le n - 1$), or only of $t'$ (if $|w| \ge n$). Hence, there cannot be a diamond in $\Aa$. Thus, $\Aa$ is unambiguous if and only if there is no $(s, t)$-path of length $n - 1$ in $G$.

\begin{figure}[ht]\centering
 \centering
\begin{tikzpicture} [node distance = 2cm]
\tikzset{every state/.style={inner sep=1pt,minimum size=1.5em}}


\node [state] at (-0.5, -1.5) (ctree1) {$f$};

\node [state] at (1, -2.25) (ctree10) {};
\node [state] at (1, -0.75) (ctree11) {};
\node [] at (1.75, -2.5) (ctree100) {};
\node [] at (1.75, -2) (ctree101) {};
\node [] at (1.75, -1) (ctree110) {};
\node [] at (1.75, -0.5) (ctree111) {};

\node [] at (2.25, -1.75) (preleaf) {};
\node [state] at (3, -2) (leaf) {};

\node [] at (2, -3) () {$T$};

\draw [draw=black] (3.5,-0.25) rectangle (0.5,-2.75);

\node [] at (2, -2.25) () {$\ldots$};
\node [] at (2, -0.75) () {$\ldots$};

\node [state] at (4.5, -2.75) (cs1) {};
\node [state] at (4.5, -0.75) (ct1) {$s'$};
\node [state] at (10.5, -1.75) (t1end) {$t'$};

\draw [draw=black] (11,-2.25) rectangle (10,-1.25);

\node [state] at (7.25, -1.25) (c2) {};

\draw [draw=black,dotted] (11,-3.25) rectangle (4,-0.25);
\node [] at (11.5, -3) (t1) {$V'$};

\node [state] at (4.5, -4) (time1) {$p_1$};
\node [state] at (6, -4) (time2) {$p_2$};
\node [] at (7.5, -4) (time3) {$\ldots$};
\node [state] at (9, -4) (time4) {$p_n$};

\draw [draw=black] (9.5,-4.5) rectangle (4,-3.5);

\path [-stealth, thick]

(t1end) edge [loop right,dashed] node[above] {} (t1end)
(t1end) edge [loop right,min distance=10mm,in=-30,out=30] node[above] {} (t1end)


(ctree1) edge [dashed] node[above] {} (ctree10)
(ctree1) edge [] node[above] {} (ctree11)

(ctree10) edge [dashed] node[above] {} (ctree100)
(ctree10) edge [] node[above] {} (ctree101)

(ctree11) edge [dashed] node[above] {} (ctree110)
(ctree11) edge [] node[above] {} (ctree111)


(time1) edge [dashed,bend right=15] node[above] {} (time2)
(time1) edge [bend left=15] node[above] {} (time2)

(time2) edge [dashed,bend right=15] node[above] {} (time3)
(time2) edge [bend left=15] node[above] {} (time3)

(time3) edge [dashed,bend right=15] node[above] {} (time4)
(time3) edge [bend left=15] node[above] {} (time4)

(ctree1) edge [dotted] node[above] {} (ct1)
(ctree1) edge [dotted,bend right=20] node[above] {} (time1)

(preleaf) edge [dashed] node[above] {} (leaf)

(leaf) edge [bend left=5] node[above] {} (c2)
(leaf) edge [dashed,bend right=5] node[above] {} (c2)
;
\end{tikzpicture}
\caption{The construction in the proof of~\Cref{thm:nl-unamb}. Transitions by letter $a$ are depicted by solid arrows, by letter $b$ by dashed arrows, and by letter $c$ by dotted arrows. All states inside of solid rectangles are mapped by $c$ to $\{f\}$.
}\label{fig:nl-unamb}
\end{figure}
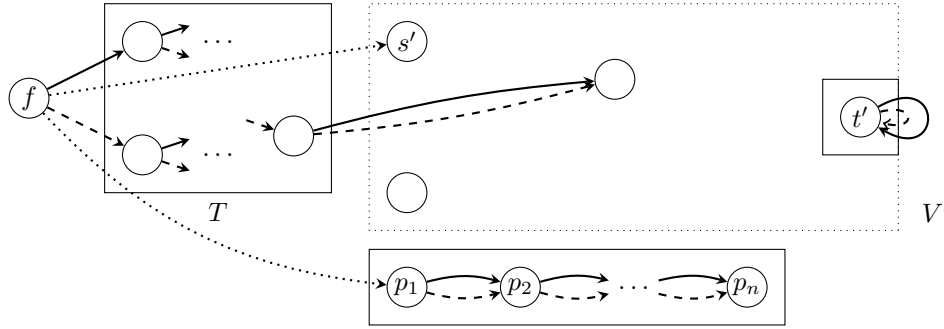

%% file: sec-2-im-b.tex
We call an NFA $\Aa = (Q, \Sigma, \Delta)$ \emph{$k$-image-bounded} if for every state $q \in Q$ and every word $w \in \Sigma^*$ we have $|q \cdot w| \le k$.
Clearly, an NFA is $1$-image-bounded if and only if it is a DFA. Hence, $2$-image-bounded NFAs can be considered as a relaxation of the notion of a DFA by adding very limited nondeterminism. In particular, not every NFA with $|q \cdot x| \le 2$ for all $q \in Q$, $x \in \Sigma$ is $2$-image-bounded, not even when there is only one nondeterministic transition (that is, when there is only one pair $q \in Q$, $x \in \Sigma$ with $|q \cdot x| = 2$). This is shown by the following NFA over a one-letter alphabet:
\raisebox{-0.25\height}{\scalebox{.5}{\begin{tikzpicture} [node distance = 2cm]
\tikzset{every state/.style={inner sep=1pt,minimum size=1.5em}}
\node [state] at (0, 0) (1) {};
\node [state] at (2, 0) (2) {};
\node [] at (4, 0) (3) {$\ldots$};
\node [state] at (6, 0) (4) {};
\path [-stealth, thick]
(1) edge [] node[above] {} (2)
(2) edge [] node[above] {} (3)
(3) edge [] node[above] {} (4)
(4) edge [bend right=12] node[above] {} (1)

(1) edge [loop left] node[above] {} (1)
;
\end{tikzpicture}}}. This is not surprising, since having only one nondeterministic transition is a property of the actions of the letters only, while being $k$-image-bounded is a property of the actions of all words. In matrix terms, for a set $\MM$ of nonnegative matrices, the corresponding NFA is $k$-image-bounded if and only if each matrix in the monoid generated by $\MM$ has at most $k$ strictly positive entries in every~row.

We also remark on the differences between $k$-image-bounded NFAs and the classical notion of $k$-ambiguous NFAs, see e.g.~\cite{Weber1991}. First, $k$-ambiguity is defined for automata with initial and final states, while $k$-image-boundedness is a property of semi-automata. Second, finite ambiguity bounds the number of possible paths starting in the same state and labelled by the same word, while $k$-image-boundedness allows for any number of such paths as long as all of them end in at most $k$ different states. What is especially important in the context of this paper, $k$-image-boundedness still makes sense for strongly connected NFAs, while a $k$-ambiguous NFA that is strongly connected must be unambiguous (that is, $1$-ambiguous).

A simple analysis of the proofs of~\Cref{thm:nl-complete} and~\Cref{thm:nl-unamb} shows that the NFAs constructed in the reductions are $2$-image-bounded (intuitively, going from left to right in the pictures can result in an image of size at most two, and going from right to left means applying $c$, which shrinks the image to a singleton or the empty set). Hence, both theorems hold true when the NFAs are promised to belong to this class (it is easy to show that deciding $2$-image-boundedness is \NL-complete, even for strongly connected NFAs). Moreover, \Cref{thm:nl-complete} provides a tight bound in this case, as shown by the following result.

\begin{proposition}\label{prop:2-im-b-in-nl}
    Deciding if a $2$-image-bounded NFA is complete is in \NL.
\end{proposition}

The proof of this proposition relies on the following simple lemma.

\begin{lemma}\label{lemma:2-im-b-completeness}
    A $2$-image-bounded NFA $\Aa = (Q, \Sigma, \Delta)$ is incomplete if and only if for every pair $p, q \in Q$ of its states there exists a word $w$ such that $p \cdot w = q \cdot w = \emptyset$.
\end{lemma}
\begin{proof}
    If $\Aa$ is incomplete, then by definition it admits a word that maps every state to the empty set.
    In the opposite direction, let $Q = \{q_1, \ldots, q_n\}$. Let $w_1$ be a word such that $q_1 \cdot w_1 = \emptyset$. Since $\Aa$ is $2$-image-bounded, we have that $|q_2 \cdot w_1| \le 2$, and by the assumption of the lemma there is $w_2$ with $q_2 \cdot w_1 w_2 = \emptyset$. Similarly, there are $w_3$ with $q_3 \cdot w_1 w_2 w_3 = \emptyset$, etc., and $w_n$ with $q_n \cdot w_1 \cdots w_n = \emptyset$.
    Define $w = w_1 \cdots w_n$.
    For all $q_i \in Q$ we have $q_i \cdot w = \emptyset$; i.e., $\Aa$ is incomplete.
\end{proof}

For a $2$-image-bounded NFA, we have $|p \cdot w \cup q \cdot w| \le 4$ for every pair $p, q \in Q$ of states and every word~$w$. Thus, \Cref{lemma:2-im-b-completeness} provides an algorithm for deciding if a given $2$-image-bounded NFA is complete in~\NL: for every pair $p, q$ of states, nondeterministically guess a word $w$ such that $p \cdot w = q \cdot w = \emptyset$. If all guesses are successful, the NFA is incomplete. This finishes the proof of~\Cref{prop:2-im-b-in-nl}. Collecting everything together, we get the following result.

\begin{theorem}
    Deciding if a binary strongly connected $2$-image-bounded unambiguous NFA is complete is \NL-complete.
\end{theorem}

%% file: sec-codes.tex
In this section, we explain a classical correspondence between strongly connected automata and regular codes, and thus obtain complexity results for the latter. This correspondence is a central concept in the book~\cite{Berstel2010}, where it is discussed in much more detail. We refer to Chapters 2-4 of this book for the proofs of the statements mentioned below.

A \emph{code} $X$ over an alphabet $\Sigma$ is a set of words such that no word over~$\Sigma$ admits more than one factorisation over $X$. That is, $X$ is a code if and only if for any words $x_1, \ldots, x_n, y_1, \ldots, y_m \in X$ with $x_1 \cdots x_n = y_1 \cdots y_m$ we have $m = n$ and $x_i = y_i$ for all~$i$.

Let $\Aa = (Q, \Sigma, \Delta)$ be a strongly connected unambiguous NFA. For a state $q \in Q$, define the \emph{set of first return words of $q$} as the set of words labelling paths that start and end in~$q$, but do not visit $q$ in-between. It is easy to see that for any $q \in Q$, the set of first return words of $q$ in $\Aa$ is a code that is also a regular language (called simply a regular code below). Conversely, for each regular code $X$ there exists a strongly connected unambiguous NFA $\Aa$ such that $X$ is the set of first return words of a state $q$. Note that $\Aa$ is not unique, but all such NFAs inherit a lot of properties of the corresponding codes. Below we discuss two such properties, completeness and synchronisation.

A code $X$ over an alphabet $\Sigma$ is called \emph{complete} if every word appears in a sequence of codewords. Formally, this means that for every word $w \in \Sigma^*$ there exist words $u, v \in \Sigma^*$ such that $uwv \in X^*$.
 
Fix a regular code $X$ and an arbitrary strongly connected unambiguous NFA $\Aa$ such that~$X$ is the set of first return words of a state $q$ in $\Aa$. Then $X$ is complete if and only if $\Aa$ is complete. Hence, it is reasonable to assume that $X$ is represented by $\Aa$ in the input. Note that if one takes $q$ as the unique initial and accepting state, the resulting NFA recognises~$X^*$. \Cref{thm:nl-complete} implies that deciding if a given regular code is complete is \NL-hard, assuming the representation discussed above.

A \emph{prefix code} $X$ over an alphabet $\Sigma$ is a set of words such that no word in $X$ is a prefix of another word in $X$. It is easy to see that every prefix code is indeed a code. A prefix code~$X$ is \emph{maximal} if for every word $w \in \Sigma^*$ such that $w$ is a prefix of a word in $X$ and $w \not \in X$, we have that for every $x \in \Sigma$, $wx$ is a prefix of a word in $X$. 
A word $w$ is  \emph{synchronising} for a maximal prefix code $X$ if every word ending with $w$ is can be factorised over $X$. Formally, this means that for every word $u \in \Sigma^*$, $uw \in X^*$. 

In the same way as for regular codes, 
for each regular maximal prefix code $X$ one can put in correspondence a strongly connected total DFA $\Aa$. We have that $X$ is synchronising if and only if $\Aa$ is synchronising. Note that $\Aa$ in this correspondence is again not unique, but every $\Aa$ whose set of first return words is $X$ satisfies this property, and thus can be taken as a representation of $X$. \Cref{thm:nl-sync} implies that deciding if a given regular maximal prefix code is synchronising is \NL-complete, assuming such a representation.

Let us also briefly remark that completeness and unambiguity of regular languages that are defined directly by automata that accept these languages are studied in~\cite{Cho1992,Huynh1992}.

%% file: sec-conclusions.tex
The main open problem left by this paper is the precise complexity of deciding  completeness for unambiguous NFAs. Recall that it is \NL-hard and belongs to \DET. While the gap between these two classes is relatively small (in the sense that few computational complexity classes lie between them), the distinction between \NL-complete and \DET-complete problems is important. Intuitively speaking, \NL-complete problems can be considered as those that admit combinatorial characterisations (as happens, for example, for completeness of $2$-image-bounded NFAs), while \DET-complete problems are characterised in terms of linear~algebra.

Another interesting open problem of the same spirit is finding subclasses of NFAs where the considered problems are not \NL-hard. Such subclasses are likely to have new interesting properties and characterisation. In particular, one can ask for subclasses where a random word of polynomial length provides a yes-certificate with high probability. This would immediately imply that the corresponding reachability problem is in the randomised counterpart of \LL for such a subclass.

The class of $k$-image-bounded NFAs seems to be worth further investigation.
The property is natural, but we did not find it mentioned in the literature.
It is easy to see that an incomplete $k$-image-bounded NFA always admits a mortal word of length at most~$n^{k + 1}$. Is this bound tight?